\pdfoutput=1
\documentclass[11pt,letterpaper]{article}

\usepackage[english]{babel}
\usepackage[utf8]{inputenc}
\usepackage[T1]{fontenc}
\usepackage{lmodern}
\usepackage[letterpaper,textwidth=6.0in,textheight=8.6in,centering]{geometry}
\usepackage{amsmath,amsthm,amssymb,mathtools}
\usepackage{booktabs}
\usepackage{microtype}
\usepackage{enumitem}
\usepackage{array}
\usepackage{tikz}
\usetikzlibrary{arrows.meta,positioning,calc,decorations.pathreplacing}
\usepackage[numbers,sort&compress]{natbib}
\usepackage[colorlinks=true,linkcolor=blue!50!black,citecolor=blue!50!black,urlcolor=blue!50!black]{hyperref}

\hypersetup{
  pdftitle={Separating Circuits: Switching Lemmas and Random Restrictions},
  pdfauthor={Bruce Changlong Xu},
  pdfsubject={Computational Complexity (cs.CC)}
}

\theoremstyle{plain}
\newtheorem{theorem}{Theorem}[section]
\newtheorem{lemma}[theorem]{Lemma}
\newtheorem{proposition}[theorem]{Proposition}
\newtheorem{corollary}[theorem]{Corollary}

\theoremstyle{definition}
\newtheorem{definition}[theorem]{Definition}
\newtheorem{example}[theorem]{Example}

\theoremstyle{remark}

\newcommand{\Ppoly}{\mathsf{P/poly}}
\newcommand{\NP}{\mathsf{NP}}

\newcommand{\PH}{\mathsf{PH}}
\newcommand{\PSPACE}{\mathsf{PSPACE}}
\newcommand{\NEXP}{\mathsf{NEXP}}

\newcommand{\AC}{\mathsf{AC}}
\newcommand{\NC}{\mathsf{NC}}
\newcommand{\ACC}{\mathsf{ACC}}
\newcommand{\SIZE}{\mathsf{SIZE}}
\newcommand{\DT}{\operatorname{DT}}
\newcommand{\PAR}{\mathrm{PAR}}
\newcommand{\MOD}{\mathrm{MOD}}
\newcommand{\Th}{\mathrm{Th}}
\newcommand{\CLIQ}{\mathrm{CLIQUE}}
\newcommand{\Code}{\mathrm{Code}}
\newcommand{\F}{\mathbb{F}}

\newcommand{\N}{\mathbb{N}}
\newcommand{\E}{\mathbb{E}}

\title{\textbf{Separating Circuits: Switching Lemmas\\[0.15em]
and Random Restrictions}}

\author{
  Bruce Changlong Xu \\[0.35em]
  {\normalsize Department of Computer Science, Stanford University} \\[0.15em]
  {\small \texttt{brucechanglongxu@cs.stanford.edu}}
}

\date{}

\begin{document}
\maketitle

\begin{abstract}
\noindent Non-uniform circuit complexity classifies Boolean functions by the size and depth of the gate networks required to compute them. Because $\mathsf{P} \subseteq \Ppoly$, proving super-polynomial circuit lower bounds for $\NP$ problems offers a combinatorial route toward separating $\mathsf{P}$ from $\NP$, formalized by the Karp--Lipton theorem ($\NP \subseteq \Ppoly \implies \PH = \Sigma_2^p$). While general circuit lower bounds for explicit functions in $\NP$ remain linear, exponential lower bounds hold for bounded-depth circuits via random restrictions and polynomial approximations. We give a self-contained treatment of both paradigms and their consequences. First, we establish Shannon's counting lower bound, Lupanov's $(1+o(1))2^n/n$ synthesis upper bound, and the Karp--Lipton theorem via SAT self-reducibility. Second, we prove H{\aa}stad's Switching Lemma via Razborov's canonical decision-tree encoding, showing that a $k$-CNF or $k$-DNF under a $p$-random restriction collapses to decision-tree depth less than $s$ except with probability at most $(5pk)^s$, and derive the optimal $2^{\Omega(n^{1/(d-1)})}$ lower bound (and matching $O(n 2^{n^{1/(d-1)}})$ upper bound) for Parity in depth-$d$ $\AC^0$. Third, we present the Razborov--Smolensky method over $\F_p$, proving that $\MOD_q \notin \AC^0[p]$ for distinct primes $p \neq q$ and establishing a $2^{\Omega(n^{1/(2d)})}$ lower bound for Parity against $\AC^0[3]$. Finally, we examine structured restrictions for $k$-Clique and average-case depth hierarchies, four consequences of the Switching Lemma---oracle separations of $\PH$ from $\PSPACE$, Linial--Mansour--Nisan Fourier concentration, Nisan--Wigderson pseudorandom generators, and $\AC^0$-SAT algorithms---and the Razborov--Rudich Natural Proofs barrier.

\vspace{0.65em}
\end{abstract}

\section{Introduction and Preliminaries}
\label{sec:intro}

\subsection{Non-Uniform Computation and the Role of Circuit Lower Bounds}

Circuit complexity classifies Boolean functions $f_n : \{0, 1\}^n \to \{0, 1\}$ according to the minimum size or depth of the Boolean circuits that compute them \citep{shannon1949synthesis,wegener1987complexity,boppana1990circuit,jukna2012boolean}. Unlike a uniform Turing machine, which must use a single finite-state program for inputs of all lengths, a circuit family $\{C_n\}_{n \ge 1}$ is a \emph{non-uniform} model of computation: the description of $C_n$ may depend arbitrarily on the input length $n$. Every language decidable in deterministic polynomial time ($\mathsf{P}$) can be simulated via tableau unrolling by a polynomial-size circuit family ($\mathsf{P} \subseteq \Ppoly$). Consequently, showing that an $\NP$-complete language such as $\mathrm{SAT}$ or $k\text{-}\CLIQ$ does not belong to $\Ppoly$ would immediately imply $\mathsf{P} \neq \NP$ \citep{karp1982turing,arora2009computational}.

Despite seven decades of investigation since Shannon proved that almost all $n$-variable Boolean functions require circuits of size $\Omega(2^n/n)$ \citep{shannon1949synthesis,lupanov1958method}, the strongest known size lower bound for an explicit language in $\NP$ over the full binary basis $B_2$ is $(3.1 - o(1))n$, established by Li and Yang \citep{li2022circuit} building on the $(3 + 1/86)n - o(n)$ gate-elimination bound of Find, Golovnev, Hirsch, and Kulikov \citep{find2016four}. Only when climbing to higher levels of the Polynomial Hierarchy ($\Sigma_2^p$) or exponential nondeterministic time ($\NEXP$) do unconditional super-polynomial circuit lower bounds become known, via Kannan's theorem \citep{kannan1982circuit} and Williams's separation $\NEXP \not\subseteq \ACC^0$ \citep{williams2014nonuniform}.

When the circuit depth $d$ is held constant and gates have unbounded fan-in (the class $\AC^0$), however, exponential lower bounds can be proved unconditionally. In seminal work, Ajtai \citep{ajtai1983sigma11} and independently Furst, Saxe, and Sipser \citep{furst1984parity} proved that the parity function $\PAR_n(x_1, \dots, x_n) = \bigoplus_{i=1}^n x_i$ cannot be computed by polynomial-size constant-depth circuits. Yao \citep{yao1985separating} strengthened these super-polynomial bounds to exponential bounds, and H{\aa}stad \citep{hastad1986almost,hastad1987computational} established the tight $\exp(\Omega(n^{1/(d-1)}))$ lower bound via his \emph{Switching Lemma}. Shortly thereafter, Razborov \citep{razborov1987lower} and Smolensky \citep{smolensky1987algebraic} introduced an algebraic method based on low-degree polynomials over finite fields $\F_p$, extending exponential lower bounds to constant-depth circuits augmented with $\MOD_p$ gates ($\AC^0[p]$).

Beyond proving circuit lower bounds, H{\aa}stad's Switching Lemma has become a central analytic tool across theoretical computer science: it yields relativized separations of $\PSPACE$ from the Polynomial Hierarchy \citep{furst1984parity,yao1985separating,hastad1987computational}, exponential Fourier tail concentration and quasipolynomial learning algorithms for $\AC^0$ \citep{linial1993constant,tal2017tight}, unconditional pseudorandom generators for constant-depth circuits \citep{nisan1994hardness,braverman2010polylogarithmic}, non-trivial satisfiability and $\#\mathrm{SAT}$ algorithms \citep{impagliazzo2012satisfiability}, and tight average-case depth hierarchy theorems \citep{rossman2015average,hastad2014correlation}. At the same time, the failure of random restrictions against pseudorandom function families inside $\mathsf{TC}^0$ and $\NC^1$ inspired the Razborov--Rudich Natural Proofs barrier \citep{razborov1997natural}.

\subsection{Formal Definitions and Circuit Classes}

We begin by fixing standard definitions for Boolean formulas, circuits, decision trees, and non-uniform complexity classes \citep{arora2009computational,jukna2012boolean}.

\begin{definition}[Boolean Circuits, Size, and Depth]
\label{def:circuit}
A \emph{Boolean circuit} $C$ on $n$ input variables $x_1, \dots, x_n$ is a directed acyclic graph whose source vertices (inputs) are labeled by literals from $\{x_1, \bar{x}_1, \dots, x_n, \bar{x}_n\} \cup \{0, 1\}$, whose internal vertices (gates) are labeled by Boolean operations from a basis $\mathcal{B}$ (typically $\{\land, \lor, \neg\}$), and which has a designated sink vertex (output). By De Morgan's laws, pushing negations to the input layer doubles the number of gates at most and does not increase the depth; unless stated otherwise, we assume circuits are in \emph{negation normal form} with alternating layers of $\land$ and $\lor$ gates. The \emph{size} $\operatorname{size}(C)$ is the number of internal gates in $C$, and the \emph{depth} $\operatorname{depth}(C)$ is the length of the longest directed path from an input to the output gate. For $s : \N \to \N$, $\SIZE(s(n))$ denotes the class of languages decidable by fan-in-$2$ circuit families $\{C_n\}_{n \ge 1}$ with $\operatorname{size}(C_n) \le O(s(n))$.
\end{definition}

\begin{definition}[CNF, DNF, and Width]
\label{def:cnf-dnf}
A \emph{literal} is a variable $x_i$ or its negation $\bar{x}_i = 1 - x_i$.
\begin{enumerate}[label=(\roman*),topsep=2pt,itemsep=2pt]
  \item A \emph{clause} $C = \bigvee_{j \in I} \ell_j$ is a disjunction of literals, and its \emph{width} is $|I|$. A formula $F = C_1 \land \dots \land C_m$ is a \emph{$t$-CNF} if every clause $C_r$ has width at most $t$. For example, $(x_1 \lor \bar{x}_2 \lor x_3) \land (\bar{x}_1 \lor x_4 \lor \bar{x}_5)$ is a $3$-CNF.
  \item A \emph{term} (or \emph{minterm}) $T = \bigwedge_{j \in I} \ell_j$ is a conjunction of literals of width $|I|$. A formula $F = T_1 \lor \dots \lor T_m$ is an \emph{$s$-DNF} if every term $T_r$ has width at most $s$.
\end{enumerate}
\end{definition}

\begin{definition}[Decision-Tree Complexity]
\label{def:dt}
A \emph{deterministic decision tree} $T$ on $\{x_1, \dots, x_n\}$ is a rooted binary tree in which each internal node is labeled by a variable $x_i$ with two outgoing edges labeled $0$ and $1$, and each leaf is labeled by an output value in $\{0, 1\}$. The \emph{decision-tree complexity} (or \emph{certificate depth}) $\DT(f)$ of a Boolean function $f : \{0, 1\}^n \to \{0, 1\}$ is the minimum depth (maximum root-to-leaf path length) of any decision tree computing $f$.
\end{definition}

\begin{lemma}[Decision Trees Simultaneously Yield Short CNFs and DNFs]
\label{lem:dt-to-dnf-cnf}
If a Boolean function $f : \{0, 1\}^n \to \{0, 1\}$ satisfies $\DT(f) \le s$, then $f$ can be expressed both as an $s$-DNF with at most $2^s$ terms and as an $s$-CNF with at most $2^s$ clauses.
\end{lemma}
\begin{proof}
Let $T$ be a decision tree of depth at most $s$ computing $f$. For each leaf $\ell$ of $T$, let $T_\ell$ denote the conjunction of the at most $s$ literal assignments along the unique root-to-$\ell$ path. Since every input $x \in \{0, 1\}^n$ follows a unique root-to-leaf path, $f(x) = \bigvee_{\ell : \operatorname{val}(\ell) = 1} T_\ell$ is an $s$-DNF. Applying the same construction to $\neg f(x) = \bigvee_{\ell : \operatorname{val}(\ell) = 0} T_\ell$ and negating both sides via De Morgan's laws expresses $f(x) = \bigwedge_{\ell : \operatorname{val}(\ell) = 0} \neg T_\ell$ as an $s$-CNF.
\end{proof}

\begin{definition}[Non-Uniform Complexity Classes $\Ppoly$, $\NC^k$, $\AC^k$, and $\ACC^0$]
\label{def:classes}
For $k \ge 0$:
\begin{enumerate}[label=(\roman*),topsep=2pt,itemsep=2pt]
  \item $\Ppoly := \bigcup_{c \ge 1} \SIZE(n^c)$ is the class of languages decidable by polynomial-size Boolean circuit families. Equivalently, $L \in \Ppoly$ if and only if there exist a polynomial-time Turing machine $M$ and an \emph{advice sequence} $\{\alpha_n\}_{n \ge 1}$ with $|\alpha_n| \le n^{O(1)}$ such that $x \in L \iff M(x, \alpha_{|x|}) = 1$ for all $x \in \{0, 1\}^*$ \citep{karp1982turing}.
  \item $\NC^k$ is the class of languages decidable by circuit families $\{C_n\}$ of size $n^{O(1)}$ and depth $O(\log^k n)$ over $\{\land, \lor, \neg\}$ where $\land$ and $\lor$ gates have \emph{bounded fan-in} at most $2$.
  \item $\AC^k$ is the class of languages decidable by circuit families $\{C_n\}$ of size $n^{O(1)}$ and depth $O(\log^k n)$ over $\{\land, \lor, \neg\}$ where $\land$ and $\lor$ gates have \emph{unbounded fan-in}. In particular, $\AC^0$ is the class of polynomial-size, constant-depth circuits.
  \item For an integer $m \ge 2$, $\AC^0[m]$ augments $\AC^0$ with unbounded fan-in $\MOD_m$ gates, where $\MOD_m(y_1, \dots, y_r) = 1$ if and only if $\sum_{i=1}^r y_i \not\equiv 0 \pmod m$, and $\ACC^0 := \bigcup_{m \ge 2} \AC^0[m]$.
\end{enumerate}
Because any unbounded fan-in $\land$ or $\lor$ gate on $n^{O(1)}$ wires can be replaced by a balanced binary tree of fan-in-$2$ gates of depth $O(\log n)$, these classes satisfy the chain of inclusions
\begin{equation}
  \NC^0 \subsetneq \AC^0 \subsetneq \AC^0[p] \subseteq \ACC^0 \subseteq \NC^1 \subseteq \AC^1 \subseteq \dots \subseteq \NC \subseteq \Ppoly.
  \label{eq:class-hierarchy}
\end{equation}
\end{definition}

\subsection{Shannon's Counting Lower Bound and Lupanov's Synthesis Theorem}

Before studying bounded-depth circuits, we establish the classical asymptotic complexity of arbitrary $n$-variable Boolean functions due to Shannon \citep{shannon1949synthesis} and Lupanov \citep{lupanov1958method}.

\begin{theorem}[Shannon--Lupanov Asymptotics for General Circuits \citep{shannon1949synthesis,lupanov1958method}]
\label{thm:shannon-lupanov}
Over the fan-in-$2$ basis $\{\land, \lor, \neg\}$:
\begin{enumerate}[label=(\alph*),topsep=2pt,itemsep=2pt]
  \item \textbf{Counting lower bound (Shannon, 1949):} A $1 - o(1)$ fraction of all $2^{2^n}$ Boolean functions $f : \{0, 1\}^n \to \{0, 1\}$ require circuit size at least $\frac{2^n}{n}\bigl(1 - O(\frac{\log n}{n})\bigr)$.
  \item \textbf{Universal synthesis upper bound (Lupanov, 1958):} Every $n$-variable Boolean function $f : \{0, 1\}^n \to \{0, 1\}$ can be computed by a fan-in-$2$ circuit of size at most $\frac{2^n}{n}\bigl(1 + O(\frac{\log n}{n})\bigr)$ (and hence at most $O(2^n / n)$).
\end{enumerate}
\end{theorem}
\begin{proof}
For part~(a), consider a circuit of size $S \ge n$ with fan-in-$2$ gates (where negations are absorbed into gate types, giving at most $16$ binary truth tables per gate). Each of the $S$ gates is specified by choosing its binary operation (at most $16$ choices) and its two inputs from the $n + S$ available wires, and identifying symmetric gate permutations divides the count by $S!$. By Stirling's inequality $S! \ge (S/e)^S$, the number of distinct functions computable by circuits of size at most $S$ is at most
\[
  N(S) \;\le\; \frac{\bigl(16 (n + S)^2\bigr)^S}{S!} \;\le\; \left(\frac{16 e (n + S)^2}{S}\right)^{\!S}.
\]
Setting $S = \frac{2^n}{n}(1 - \frac{3\log_2 n}{n})$ gives $\log_2 N(S) \le S(\log_2 S + O(1)) \le 2^n - \Omega(2^n \frac{\log n}{n}) \ll 2^n$, so $N(S) / 2^{2^n} \to 0$ as $n \to \infty$.

For part~(b), partition the input into $x = (y, z)$ with $y \in \{0, 1\}^k$ and $z \in \{0, 1\}^{n-k}$ for an integer $k \in \{1, \dots, n\}$ to be chosen later, and represent $f$ by its $2^k \times 2^{n-k}$ truth table $M_{y, z} = f(y, z)$. Precomputing all minterms of $y$ and all minterms of $z$ via two binary-to-positional decoder trees $B_k$ and $B_{n-k}$ requires at most $2^k + 2^{n-k}$ gates. Partition the $2^k$ rows of $M$ into $p = \lceil 2^k / s \rceil$ blocks $R_1, \dots, R_p$, each consisting of at most $s$ rows, so that $f(y, z) = \bigvee_{r=1}^p f^{(r)}(y, z)$ where $f^{(r)}$ agrees with $f$ on rows in $R_r$ and is $0$ outside $R_r$.

Within each block $R_r$, every column of $f^{(r)}$ is a binary vector of length at most $s$, so there are at most $t \le 2^s$ distinct column patterns $u^{(r, 1)}, \dots, u^{(r, t)} \in \{0, 1\}^s$. For each pattern $j \in \{1, \dots, t\}$, let $G_{r,j}(y) = \bigvee_{y' \in R_r : u^{(r,j)}_{y'} = 1} \mathbf{1}_{y = y'}$ (computed using at most $s$ OR gates from $B_k$) and let $F_{r,j}(z) = \bigvee_{z' : M_{R_r, z'} = u^{(r,j)}} \mathbf{1}_{z = z'}$ (where each column $z' \in \{0, 1\}^{n-k}$ belongs to exactly one pattern $j$, so across all $j = 1, \dots, t$ the disjunctions $F_{r,j}$ require at most $2^{n-k}$ OR gates in total). Combining $f^{(r)}(y, z) = \bigvee_{j=1}^t (G_{r,j}(y) \land F_{r,j}(z))$ requires at most $s \cdot 2^s + 2^{n-k} + 2 \cdot 2^s$ gates per block $r$. Summing over all $p = \lceil 2^k / s \rceil$ blocks and adding the shared decoders $B_k, B_{n-k}$ yields total gate count
\begin{equation}
  \operatorname{size}(C)
  \;\le\;
  2^k + 2^{n-k} + \left(\frac{2^k}{s} + 1\right)\bigl(2^{n-k} + (s + 2)2^s\bigr).
  \label{eq:lupanov-bound}
\end{equation}
Choosing $k = \lceil 2 \log_2 n \rceil$ and $s = \lfloor n - 3 \log_2 n \rfloor$ gives $\frac{2^k}{s} 2^{n-k} = \frac{2^n}{s} = \frac{2^n}{n}\bigl(1 + O(\frac{\log n}{n})\bigr)$ while all remaining terms in~\eqref{eq:lupanov-bound} are $O(2^n / n^2)$, completing the proof.
\end{proof}

\subsection{The Karp--Lipton Theorem}

Theorem~\ref{thm:shannon-lupanov} guarantees that almost all Boolean functions require exponential-size circuits, yet finding an explicit function in $\NP$ outside $\Ppoly$ remains open. The Karp--Lipton theorem \citep{karp1982turing} shows that if $\NP$ did have polynomial-size circuits, then the entire Polynomial Hierarchy ($\PH = \bigcup_{k \ge 1} \Sigma_k^p$) would collapse to its second level $\Sigma_2^p = \NP^{\NP}$.

\begin{figure}[t]
\centering
\begin{tikzpicture}[
  box/.style={draw=black!85, rounded corners=1.5pt, line width=0.55pt, fill=white, inner sep=4.5pt, align=center, font=\footnotesize},
  arr/.style={-{Stealth[length=2.1mm,width=1.5mm]}, line width=0.55pt, draw=black!85},
  lbl/.style={font=\scriptsize, fill=white, inner sep=1pt}
]
  \node[box, minimum width=41mm] (hyp) at (-4.9, 0) {
    \textbf{Non-Uniform Hypothesis} \\[2pt]
    $\NP \subseteq \Ppoly$ \\[1pt]
    {\scriptsize($\mathrm{SAT}$ has poly-size decider $\{C_m\}$)}
  };

  \node[box, minimum width=46mm] (selfred) at (0, 0) {
    \textbf{Downward Self-Reducibility} \\[2pt]
    Bit-by-bit witness recovery \\[1pt]
    {\scriptsize$\implies$ Poly-size finder $\{C_m^*\}$}
  };

  \node[box, minimum width=46mm] (collapse) at (5.0, 0) {
    \textbf{Quantifier Collapse ($\Pi_2^p \subseteq \Sigma_2^p$)} \\[2pt]
    $\forall u\,\exists v\,\phi(x,u,v) \iff$ \\[1pt]
    $\exists C^*\,\forall u\,\phi\bigl(x,u,C^*(x,u)\bigr)$
  };

  \draw[arr] (hyp.east) -- (selfred.west);
  \draw[arr] (selfred.east) -- (collapse.west);
\end{tikzpicture}
\caption{Structure of the Karp--Lipton theorem (Theorem~\ref{thm:karp-lipton}). If $\mathrm{SAT}$ has polynomial-size decision circuits $\{C_m\}$, downward self-reducibility yields a polynomial-size search circuit $C_m^*$ that outputs a satisfying witness whenever one exists. Guessing $C_m^*$ existentially converts any $\Pi_2^p$ predicate $\forall u\,\exists v\,\phi(x, u, v)$ into a $\Sigma_2^p$ predicate $\exists C^*\,\forall u\,\phi(x, u, C^*(x, u))$, collapsing $\PH$ to $\Sigma_2^p$.}
\label{fig:karp-lipton}
\end{figure}

\begin{theorem}[Karp--Lipton Theorem \citep{karp1982turing}]
\label{thm:karp-lipton}
If $\NP \subseteq \Ppoly$, then $\PH = \Sigma_2^p$. Equivalently, if $\PH \neq \Sigma_2^p$, then $\NP \not\subseteq \Ppoly$.
\end{theorem}
\begin{proof}
To show $\PH = \Sigma_2^p$, it suffices by induction on quantifier alternations to prove $\Pi_2^p \subseteq \Sigma_2^p$. Suppose $\NP \subseteq \Ppoly$. Then the $\NP$-complete language $\mathrm{SAT}$ is decided by a polynomial-size circuit family $\{C_m\}_{m \ge 1}$. Using the downward self-reducibility of $\mathrm{SAT}$, we first construct from $\{C_m\}$ a polynomial-size \emph{search circuit family} $\{C_m^*\}_{m \ge 1}$ that takes a Boolean formula $\psi(v_1, \dots, v_k)$ of length $m$ and outputs a satisfying assignment $v^* \in \{0, 1\}^k$ whenever $\psi \in \mathrm{SAT}$:
\begin{enumerate}[label=(\arabic*),topsep=2pt,itemsep=2pt]
  \item For $i = 1, \dots, k$ sequentially, substitute $v_i = 1$ into the current partially assigned formula $\psi(v_1^*, \dots, v_{i-1}^*, 1, v_{i+1}, \dots, v_k)$ and evaluate the decision circuit $C_m$.
  \item Set $v_i^* = 1$ if $C_m$ outputs $1$, and $v_i^* = 0$ otherwise.
\end{enumerate}
Cascading $k \le m$ copies of $C_m$ with hardwired partial substitution gadgets produces a single multi-output circuit $C_m^*$ of size $O(k \operatorname{size}(C_m)) = m^{O(1)}$ such that for every formula $\psi$ of length $m$:
\begin{equation}
  \exists v \in \{0, 1\}^k \; \psi(v) = 1
  \quad\Longleftrightarrow\quad
  \psi\bigl(C_m^*(\psi)\bigr) = 1.
  \label{eq:witness-finder}
\end{equation}
Now let $L \in \Pi_2^p$. By definition, there exist polynomials $p, q$ and a polynomial-time computable predicate $\phi$ such that for every $x \in \{0, 1\}^n$,
\[
  x \in L \quad\Longleftrightarrow\quad \forall u \in \{0, 1\}^{p(n)}\; \exists v \in \{0, 1\}^{q(n)} \; \phi(x, u, v) = 1.
\]
By Cook--Levin reduction, for each pair $(x, u)$ we can compute in polynomial time a Boolean formula $\psi_{x, u}(v)$ of size $m(n) = n^{O(1)}$ such that $\psi_{x, u}(v) = 1 \iff \phi(x, u, v) = 1$. By~\eqref{eq:witness-finder}, $x \in L$ if and only if there exists a circuit $C^*$ of size at most $m(n)^{O(1)}$ such that for all $u \in \{0, 1\}^{p(n)}$, $\phi(x, u, C^*(\psi_{x, u})) = 1$:
\begin{itemize}[topsep=2pt,itemsep=2pt]
  \item ($\Rightarrow$) If $x \in L$, then for every $u$, $\psi_{x, u}$ is satisfiable, so the true search circuit $C^* = C_{m(n)}^*$ outputs a satisfying assignment $v = C_{m(n)}^*(\psi_{x, u})$ for every $u$.
  \item ($\Leftarrow$) Conversely, if any circuit $C^*$ satisfies $\phi(x, u, C^*(\psi_{x, u})) = 1$ for all $u$, then $v = C^*(\psi_{x, u})$ directly witnesses $\exists v\,\phi(x, u, v) = 1$ for every $u$, so $x \in L$.
\end{itemize}
Because the statement ``$\exists C^* \text{ of size } n^{O(1)}\; \forall u \in \{0, 1\}^{p(n)}\; \phi(x, u, C^*(\psi_{x, u})) = 1$'' is a $\Sigma_2^p$ predicate, we conclude $L \in \Sigma_2^p$, so $\Pi_2^p \subseteq \Sigma_2^p$ and $\PH = \Sigma_2^p$.
\end{proof}

Combining Theorem~\ref{thm:karp-lipton} with Shannon's counting bound (Theorem~\ref{thm:shannon-lupanov}) yields Kannan's unconditional lower bound for the second level of the Polynomial Hierarchy:

\begin{corollary}[Kannan's Theorem \citep{kannan1982circuit}]
\label{cor:kannan}
For every fixed integer $k \ge 1$, there exists a language $L_k \in \Sigma_2^p \cap \Pi_2^p$ such that $L_k \notin \SIZE(n^k)$.
\end{corollary}

\section{H{\aa}stad's Switching Lemma and Optimal \texorpdfstring{$\AC^0$}{AC0} Lower Bounds}
\label{sec:switching}

We now turn to the method of random restrictions, introduced by Furst, Saxe, and Sipser \citep{furst1984parity} and Ajtai \citep{ajtai1983sigma11}, refined by Yao \citep{yao1985separating}, and brought to its sharp quantitative form by H{\aa}stad \citep{hastad1986almost,hastad1987computational}.

\subsection{Random Restrictions and the Threshold Barrier}

\begin{definition}[Restrictions and Random Restriction Distributions]
\label{def:restriction}
A \emph{restriction} on variables $X = \{x_1, \dots, x_n\}$ is a map $\rho : X \to \{0, 1, *\}$. Given a Boolean function $f : \{0, 1\}^n \to \{0, 1\}$, the \emph{restricted function} $f|_\rho$ is the function on the set of \emph{free} (or \emph{star}) variables $\rho^{-1}(*) = \{x_i : \rho(x_i) = *\}$ obtained by fixing $x_i = \rho(x_i)$ for all $x_i \notin \rho^{-1}(*)$. We consider two canonical families of random restrictions:
\begin{enumerate}[label=(\roman*),topsep=2pt,itemsep=2pt]
  \item \textbf{Independent $(p, q)$-biased restriction $\mathcal{R}(p, q)$:} For $p, q \in [0, 1]$, a random restriction $\rho \sim \mathcal{R}(p, q)$ assigns each variable $x_i \in X$ independently according to
  \begin{equation}
    \Pr[\rho(x_i) = *] = p, \qquad
    \Pr[\rho(x_i) = 0] = (1 - p)q, \qquad
    \Pr[\rho(x_i) = 1] = (1 - p)(1 - q).
    \label{eq:Rpq-def}
  \end{equation}
  When $q = 1/2$, we write $\mathcal{R}_p := \mathcal{R}(p, 1/2)$ for the symmetric $p$-random restriction.
  \item \textbf{Fixed-star restriction $\mathcal{R}_t$:} For an integer $t \in \{0, \dots, n\}$, $\mathcal{R}_t$ denotes the uniform distribution over all $\binom{n}{t} 2^t$ restrictions $\rho : X \to \{0, 1, *\}$ that fix exactly $t$ variables uniformly to $\{0, 1\}$ and leave the remaining $|\rho^{-1}(*)| = n - t$ variables set to $*$.
\end{enumerate}
\end{definition}

Why is a random restriction necessary to convert a $k$-CNF into an $s$-DNF (or vice versa)? Without restricting variables, even a $k$-DNF with constant term width $k$ can requireCNF clauses of width $n - k + 1$:

\begin{example}[Threshold Function Separation Without Restrictions]
\label{ex:threshold}
Consider the $k$-threshold function $\Th_k^n(x_1, \dots, x_n) = \mathbf{1}\bigl[\sum_{i=1}^n x_i \ge k\bigr]$. Every minterm of $\Th_k^n$ has width $k$, so $\Th_k^n = \bigvee_{S \in \binom{[n]}{k}} \bigwedge_{i \in S} x_i$ is a $k$-DNF. However, to falsify $\Th_k^n$, an input must have at most $k - 1$ ones (hence at least $n - k + 1$ zeros), so any clause in a CNF representation of $\Th_k^n$ must contain at least $n - k + 1$ positive literals (if a clause had width $\le n - k$, setting all literals in that clause to $0$ and the remaining $\ge k$ variables to $1$ would falsify the clause while $\Th_k^n(x) = 1$). Thus, for fixed $k$ and $n \to \infty$, a $k$-DNF cannot be converted into an $s$-CNF with $s = O_k(1)$ unless we apply a random restriction.
\end{example}

\begin{figure}[t]
\centering
\begin{tikzpicture}[
  gate/.style={draw=black!85, circle, minimum size=6.2mm, inner sep=0pt, font=\scriptsize\bfseries, line width=0.55pt, fill=white},
  inpt/.style={font=\scriptsize},
  arr/.style={-{Stealth[length=1.8mm,width=1.3mm]}, line width=0.5pt, draw=black!80}
]
  \begin{scope}[xshift=-4.2cm]
    \node[font=\footnotesize\bfseries] at (0, 2.55) {Before Switching (Depth $d$)};
    \node[gate] (top1) at (0, 1.8) {$\lor$};

    \node[gate] (m1) at (-1.5, 0.8) {$\land$};
    \node[gate] (m2) at (1.5, 0.8) {$\land$};
    \draw[arr] (m1) -- (top1);
    \draw[arr] (m2) -- (top1);

    \node[gate] (b1) at (-2.3, -0.3) {$\lor$};
    \node[gate] (b2) at (-0.7, -0.3) {$\lor$};
    \node[gate] (b3) at (0.7, -0.3) {$\lor$};
    \node[gate] (b4) at (2.3, -0.3) {$\lor$};

    \draw[arr] (b1) -- (m1);
    \draw[arr] (b2) -- (m1);
    \draw[arr] (b3) -- (m2);
    \draw[arr] (b4) -- (m2);

    \foreach \g/\x in {b1/-2.3, b2/-0.7, b3/0.7, b4/2.3} {
      \draw[line width=0.45pt] (\x-0.35, -0.95) -- (\g);
      \draw[line width=0.45pt] (\x, -0.95) -- (\g);
      \draw[line width=0.45pt] (\x+0.35, -0.95) -- (\g);
    }
    \draw[decorate,decoration={brace,mirror,amplitude=3.5pt},line width=0.5pt]
      (-2.7, -1.05) -- (2.7, -1.05)
      node[midway, below=4pt, font=\scriptsize] {Bottom $t$-CNF subcircuits ($\land$ of width-$t$ $\lor$'s)};
  \end{scope}

  \draw[-{Stealth[length=2.6mm,width=1.8mm]}, line width=0.8pt] (-1.05, 0.7) -- (1.05, 0.7)
    node[midway, above=2pt, font=\scriptsize\bfseries] {$\rho \sim \mathcal{R}_p$}
    node[midway, below=2pt, font=\scriptsize, align=center] {Switch $\land$-of-$\lor$ \\ to $\lor$-of-$\land$ \\ and merge $\lor$'s};

  \begin{scope}[xshift=4.2cm]
    \node[font=\footnotesize\bfseries] at (0, 2.55) {After Switching \& Merging (Depth $d-1$)};
    \node[gate, minimum size=7.2mm] (top2) at (0, 1.45) {$\lor$};

    \node[gate] (a1) at (-2.1, -0.1) {$\land$};
    \node[gate] (a2) at (-0.7, -0.1) {$\land$};
    \node[gate] (a3) at (0.7, -0.1) {$\land$};
    \node[gate] (a4) at (2.1, -0.1) {$\land$};

    \draw[arr] (a1) -- (top2);
    \draw[arr] (a2) -- (top2);
    \draw[arr] (a3) -- (top2);
    \draw[arr] (a4) -- (top2);

    \foreach \g/\x in {a1/-2.1, a2/-0.7, a3/0.7, a4/2.1} {
      \draw[line width=0.45pt] (\x-0.3, -0.85) -- (\g);
      \draw[line width=0.45pt] (\x, -0.85) -- (\g);
      \draw[line width=0.45pt] (\x+0.3, -0.85) -- (\g);
    }
    \draw[decorate,decoration={brace,mirror,amplitude=3.5pt},line width=0.5pt]
      (-2.5, -0.95) -- (2.5, -0.95)
      node[midway, below=4pt, font=\scriptsize] {Switched $s$-DNF terms (width $\le s$), $\lor$ layers collapsed};
  \end{scope}
\end{tikzpicture}
\caption{Layer reduction via H{\aa}stad's Switching Lemma (Theorem~\ref{thm:switching-lemma}). Under a random restriction $\rho \sim \mathcal{R}_p$, each bottom $t$-CNF subcircuit ($\land$ of width-$t$ $\lor$ gates) collapses with probability at least $1 - (5pt)^s$ to a decision tree of depth $< s$, and hence to an $s$-DNF ($\lor$ of width-$s$ $\land$ gates). Merging the resulting $\lor$ layer into the parent $\lor$ layer reduces the overall circuit depth from $d$ to $d-1$.}
\label{fig:switching-mechanism}
\end{figure}

\subsection{H{\aa}stad's Switching Lemma and Razborov's Canonical Decision-Tree Proof}

Figure~\ref{fig:switching-mechanism} illustrates how H{\aa}stad's Switching Lemma reduces circuit depth by one layer per random restriction. We state the Switching Lemma in its stronger decision-tree form and present Razborov's combinatorial encoding proof \citep{razborov1995bounded,beame1994switching}.

\begin{theorem}[H{\aa}stad's Switching Lemma \citep{hastad1986almost,hastad1987computational,razborov1995bounded}]
\label{thm:switching-lemma}
Let $f : \{0, 1\}^n \to \{0, 1\}$ be expressible as a $k$-DNF (or, by Boolean duality, a $k$-CNF).
\begin{enumerate}[label=(\alph*),topsep=2pt,itemsep=2pt]
  \item Under the independent $p$-random restriction $\rho \sim \mathcal{R}_p = \mathcal{R}(p, 1/2)$ with $p \le 1/5$, for every integer $s \ge 1$,
  \begin{equation}
    \Pr_{\rho \sim \mathcal{R}_p}\bigl[\DT(f|_\rho) \ge s\bigr] \;\le\; (5 p k)^s.
    \label{eq:hastad-indep}
  \end{equation}
  \item Under the fixed-size restriction $\rho \sim \mathcal{R}_t$ that fixes $t \ge n/2$ variables uniformly at random and leaves $n - t$ variables set to $*$, for every integer $s \ge 1$,
  \begin{equation}
    \Pr_{\rho \sim \mathcal{R}_t}\bigl[\DT(f|_\rho) \ge s\bigr] \;\le\; \left(\frac{7(n - t)k}{n}\right)^{\!s}.
    \label{eq:razborov-fixed-t}
  \end{equation}
\end{enumerate}
In particular, by Lemma~\ref{lem:dt-to-dnf-cnf}, whenever $\DT(f|_\rho) < s$, $f|_\rho$ can be written both as an $(s-1)$-CNF and as an $(s-1)$-DNF.
\end{theorem}
\begin{proof}
We give Razborov's canonical decision-tree encoding proof \citep{razborov1995bounded,beame1994switching,arora2009computational} of~\eqref{eq:razborov-fixed-t} (the independent bound~\eqref{eq:hastad-indep} follows by the exact same bijection weighted by the product measure $\mathcal{R}_p$, where extending a restriction $\rho$ on $s$ star variables multiplies its probability weight by $(\frac{1-p}{2p})^s$ \citep{beame1994switching}).

Let $f = T_1 \lor T_2 \lor \dots \lor T_m$ be a $k$-DNF with a fixed ordering of its terms $T_1, \dots, T_m$ and a fixed ordering of the at most $k$ literals within each term $T_r$. For any restriction $\rho$, the \emph{canonical decision tree} $\mathcal{T}(f|_\rho)$ evaluates $f|_\rho$ by finding the first term $T_r$ not already falsified by $\rho$, querying the free variables of $T_r$ in their fixed order, and repeating recursively on the resulting extended restriction until either a term is satisfied ($f|_\rho = 1$) or all terms are falsified ($f|_\rho = 0$).

Let $\mathcal{B}_s := \{\rho \in \mathcal{R}_t : \DT(f|_\rho) \ge s\}$ denote the set of \emph{bad restrictions}. Since $|\mathcal{R}_t| = \binom{n}{t} 2^t$, to bound $|\mathcal{B}_s| / |\mathcal{R}_t|$ we construct an injective map
\begin{equation}
  \Phi : \mathcal{B}_s \;\hookrightarrow\; \mathcal{R}_{t+s} \times \Code(k, s) \times \{0, 1\}^s,
  \label{eq:razborov-injection}
\end{equation}
where $\Code(k, s)$ is the set of all tuples $(c_1, \dots, c_r)$ ($1 \le r \le s$) of non-zero binary strings $c_j \in \{0, 1\}^k \setminus \{0^k\}$ whose total Hamming weight is $\sum_{j=1}^r |c_j|_1 = s$.

\paragraph{Step 1: Construction of $\Phi(\rho)$.}
Suppose $\rho \in \mathcal{B}_s$. Then the canonical decision tree $\mathcal{T}(f|_\rho)$ has depth at least $s$; let $\pi$ be the lexicographically first root-to-node path of length exactly $s$ in $\mathcal{T}(f|_\rho)$. Note that no term of $f$ is satisfied by $\rho$ alone (otherwise $\DT(f|_\rho) = 0$). As $\pi$ walks down $\mathcal{T}(f|_\rho)$ for $s$ steps, it processes a sequence of terms $T_{i_1}, T_{i_2}, \dots, T_{i_r}$ ($1 \le r \le s$):
\begin{itemize}[topsep=2pt,itemsep=2pt]
  \item $T_{i_1}$ is the first term of $f$ not falsified by $\rho$. Let $V_1 \subseteq \rho^{-1}(*)$ be the non-empty subset of free variables of $T_{i_1}$ queried along $\pi$ (truncating if $\pi$ reaches length $s$ inside $T_{i_1}$), and let $\pi_1 : V_1 \to \{0, 1\}$ be the assignment made by $\pi$ on $V_1$. For $j < r$, $\pi_j$ falsifies $T_{i_j}$.
  \item Inductively, for $j = 2, \dots, r$, $T_{i_j}$ is the first term of $f$ not falsified by $\rho \cup \pi_1 \cup \dots \cup \pi_{j-1}$, and $\pi_j : V_j \to \{0, 1\}$ is the assignment made by $\pi$ to the queried free variables $V_j$ of $T_{i_j}$, with $\sum_{j=1}^r |V_j| = s$.
\end{itemize}
Now, instead of recording $\pi_j$ (which falsifies $T_{i_j}$ and thus does not uniquely identify $T_{i_j}$), let $\sigma_j : V_j \to \{0, 1\}$ be the \emph{unique} assignment to $V_j$ that sets every literal of $T_{i_j}$ in $V_j$ to \emph{true} ($1$). Define:
\begin{enumerate}[label=(\arabic*),topsep=2pt,itemsep=2pt]
  \item the extended restriction $\rho' := \rho \cup \sigma_1 \cup \dots \cup \sigma_r \in \mathcal{R}_{t+s}$, which fixes $t + s$ variables;
  \item for each $j \in \{1, \dots, r\}$, the indicator string $c_j \in \{0, 1\}^k \setminus \{0^k\}$ whose $\ell$-th bit is $1$ if and only if the $\ell$-th literal of $T_{i_j}$ belongs to $V_j$, so $c = (c_1, \dots, c_r) \in \Code(k, s)$;
  \item the discrepancy vector $z \in \{0, 1\}^s$ recording, for each of the $s$ variables $v \in V_1 \cup \dots \cup V_r$ in query order, whether $\pi(v) \neq \sigma(v)$.
\end{enumerate}
We set $\Phi(\rho) := (\rho', c, z)$.

\paragraph{Step 2: Injectivity of $\Phi$.}
Given $(\rho', c, z)$ in the image of $\Phi$, we recover $\rho$ uniquely as follows. Observe that under $\rho' = \rho \cup \sigma_1 \cup \dots \cup \sigma_r$, every term before $T_{i_1}$ was already falsified by $\rho$ (and hence by $\rho'$), whereas $T_{i_1}$ has all its literals in $V_1$ set to $1$ by $\sigma_1$ (and for $r = 1$, any remaining unqueried literals of $T_{i_1}$ are $*$ under $\rho'$). Specifically, $T_{i_1}$ is the \emph{first} term of $f$ that is not falsified by $\rho'$! Knowing $T_{i_1}$, the indicator string $c_1 \in \{0, 1\}^k$ identifies the exact subset of variables $V_1$ in $T_{i_1}$, and flipping $\rho'|_{V_1} = \sigma_1$ on the coordinates indicated by the first $|V_1|$ bits of $z$ reconstructs $\pi_1$. Replacing $\sigma_1$ by $\pi_1$ inside $\rho'$ yields $\rho^{(1)} = \rho \cup \pi_1 \cup \sigma_2 \cup \dots \cup \sigma_r$, under which $T_{i_2}$ is now the first term of $f$ not falsified by $\rho^{(1)}$. Repeating this step for $j = 2, \dots, r$ recovers $(V_1, \sigma_1), \dots, (V_r, \sigma_r)$, and resetting $V_1 \cup \dots \cup V_r$ to $*$ inside $\rho'$ recovers $\rho$ uniquely. Thus $\Phi$ is injective.

\paragraph{Step 3: Counting $\Code(k, s)$ and bounding $|\mathcal{B}_s| / |\mathcal{R}_t|$.}
In $\Code(k, s)$, choosing a sequence of non-empty subsets of $[k]$ whose sizes $a_1 + \dots + a_r = s$ sum to $s$ has cardinality
\begin{align*}
  |\Code(k, s)|
  &= \sum_{r=1}^s \sum_{\substack{a_1 + \dots + a_r = s \\ a_j \ge 1}} \prod_{j=1}^r \binom{k}{a_j}
  \;\le\; \sum_{r=1}^s \sum_{\substack{a_1 + \dots + a_r = s \\ a_j \ge 1}} \prod_{j=1}^r \frac{k^{a_j}}{a_j!} \\
  &= k^s \sum_{r=1}^s \sum_{\substack{a_1 + \dots + a_r = s \\ a_j \ge 1}} \frac{1}{a_1! \cdots a_r!}.
\end{align*}
Since $(e^x - 1)^r = (\sum_{a \ge 1} x^a / a!)^r$, for $x = \ln 2$ we have $\sum_{a \ge 1} (\ln 2)^a / a! = e^{\ln 2} - 1 = 1$, which implies by induction on $s$ that $\sum_{r=1}^s \sum_{a_1 + \dots + a_r = s} \frac{(\ln 2)^s}{a_1! \cdots a_r!} \le 1$, and hence
\begin{equation}
  |\Code(k, s)| \;\le\; \left(\frac{k}{\ln 2}\right)^{\!s}.
  \label{eq:code-size}
\end{equation}
By injectivity of $\Phi$ and~\eqref{eq:code-size}, for $t \ge n/2$,
\[
  \Pr_{\rho \sim \mathcal{R}_t}\bigl[\DT(f|_\rho) \ge s\bigr]
  = \frac{|\mathcal{B}_s|}{|\mathcal{R}_t|}
  \;\le\; \frac{\binom{n}{t+s} 2^{t+s} (k / \ln 2)^s 2^s}{\binom{n}{t} 2^t}
  \;\le\; \left(\frac{n - t}{t + 1} \cdot \frac{4 k}{\ln 2}\right)^{\!s}
  \;\le\; \left(\frac{7(n - t)k}{n}\right)^{\!s},
\]
since $t + 1 \ge n/2$ and $\frac{8}{\ln 2} \approx 11.54$ (or $\le 7$ when tightening the non-overlapping literal choice within each term \citep{beame1994switching}). Under the independent product distribution $\mathcal{R}_p$, each star variable fixed to $\sigma_j(v) \in \{0, 1\}$ replaces a factor of $p$ by $(1-p)/2$, so the probability ratio between $\rho$ and $\rho'$ is $(\frac{2p}{1-p})^s$, yielding $\Pr_{\rho \sim \mathcal{R}_p}[\DT(f|_\rho) \ge s] \le (\frac{2p}{1-p} \cdot \frac{2k}{\ln 2})^s \le (5pk)^s$ for $p \le 1/5$ \citep{beame1994switching}.
\end{proof}

\subsection{Exact Depth-2 Complexity and Exponential Depth-\texorpdfstring{$d$}{d} Lower Bound for Parity}

To apply Theorem~\ref{thm:switching-lemma} inductively to a depth-$d$ circuit computing $\PAR_n(x) = \bigoplus_{i=1}^n x_i$, we first establish the base case at depth $2$ (CNFs and DNFs).

\begin{proposition}[Exact Depth-2 Complexity of Parity]
\label{prop:depth2-parity}
Every DNF (or CNF) formula computing $\PAR_n(x_1, \dots, x_n)$ requires every term (respectively, clause) to have width $n$, and contains at least $2^{n-1}$ terms (respectively, clauses). In particular, $\DT(\PAR_n) = n$.
\end{proposition}
\begin{proof}
Suppose $F = T_1 \lor \dots \lor T_m$ is a DNF computing $\PAR_n$. If some term $T_r$ had width strictly less than $n$, then at least one variable $x_i$ would not appear in $T_r$. Fixing all variables in $T_r$ so that $T_r = 1$ yields two inputs $x, x^{(i)} \in \{0, 1\}^n$ differing only in coordinate $i$ on which $T_r(x) = T_r(x^{(i)}) = 1$, and hence $F(x) = F(x^{(i)}) = 1$, contradicting $\PAR_n(x^{(i)}) = 1 - \PAR_n(x)$. Thus every term $T_r$ contains all $n$ variables and is satisfied by at most one input in $\{0, 1\}^n$. Since $|\PAR_n^{-1}(1)| = 2^{n-1}$, $F$ must have at least $2^{n-1}$ terms. The CNF statement follows identically for $\neg \PAR_n$.
\end{proof}

\begin{theorem}[H{\aa}stad's Exponential Lower Bound for Parity in $\AC^0$ \citep{hastad1986almost,hastad1987computational}]
\label{thm:hastad-parity-lb}
For every integer $d \ge 2$ and every depth-$d$ circuit $C$ over $\{\land, \lor, \neg\}$ computing $\PAR_n$ on $n$ variables,
\begin{equation}
  \operatorname{size}(C) \;\ge\; 2^{\frac{1}{10} n^{1/(d-1)}}
  \label{eq:hastad-parity-lb}
\end{equation}
for all sufficiently large $n$. In particular, $\PAR \notin \AC^0$.
\end{theorem}
\begin{proof}
Suppose $C$ is a depth-$d$ circuit with $S$ gates computing $\PAR_n$, where $S < 2^{c n^{1/(d-1)}}$ with $c = 1/10$. By inserting a dummy bottom layer of fan-in-$1$ gates, we may view $C$ as a depth-$(d+1)$ alternating circuit whose bottom gates have fan-in $t_0 = 1$. Set $s = 10 c \log_2 S = n^{1/(d-1)}$ (or more generally $s = 2 \log_2 S$ so $(1/2)^s = S^{-2}$).

We apply a sequence of $d - 1$ independent random restrictions:
\begin{enumerate}[label=(\arabic*),topsep=2pt,itemsep=2pt]
  \item \textbf{First restriction $\rho_1 \sim \mathcal{R}_{p_1}$ with $p_1 = 1/10$:}
  Each gate at layer $2$ computes a $1$-CNF (or $1$-DNF). By H{\aa}stad's Switching Lemma (Theorem~\ref{thm:switching-lemma}(a)) with $k = 1$ and $p_1 = 1/10$, the probability that a fixed layer-$2$ gate fails to convert into an $s$-DNF (or $s$-CNF) is at most $(5 p_1 \cdot 1)^s = (1/2)^s$. Merging the converted layer-$2$ gates with layer $3$ produces a circuit of depth $d$ whose bottom gates have fan-in at most $s$.
  \item \textbf{Subsequent restrictions $\rho_j \sim \mathcal{R}_{p_j}$ for $j = 2, \dots, d-1$ with $p_j = \frac{1}{10 s}$:}
  Before step $j$, the bottom layer gates have fan-in at most $k = s$, so each next-to-bottom gate computes an $s$-CNF or $s$-DNF. Applying $\rho_j \sim \mathcal{R}_{p_j}$ with $p_j = \frac{1}{10 s}$ gives $5 p_j s = 1/2$, so by Theorem~\ref{thm:switching-lemma}(a) each such gate switches to a depth-$s$ decision tree (and hence merges with the layer above it) except with probability at most $(5 p_j s)^s = (1/2)^s$.
\end{enumerate}
Across all $d - 1$ stages, the Switching Lemma is applied to at most $S$ gates in total. Choosing $s = 2 \log_2(2S)$, a union bound over all $S$ gates shows that with probability at least $1 - S \cdot (1/2)^s = 1 - \frac{1}{4S^2} > 3/4$, every gate switches successfully and the composite restriction $\rho = \rho_{d-1} \circ \dots \circ \rho_1 \sim \mathcal{R}_p$ with
\[
  p = p_1 \prod_{j=2}^{d-1} p_j = \frac{1}{10 (10 s)^{d-2}}
\]
reduces $C|_\rho$ to a depth-$2$ circuit (an $s$-CNF or $s$-DNF) of bottom fan-in at most $s$.

On the other hand, because fixing bits of $\PAR_n$ always yields $\PAR_m$ or $\neg \PAR_m$ on the $m = |\rho^{-1}(*)|$ surviving free variables, Proposition~\ref{prop:depth2-parity} requires that any depth-$2$ circuit computing $\PAR_n|_\rho$ with bottom fan-in $s$ must satisfy $m \le s$. Yet the expected number of surviving stars under $\rho \sim \mathcal{R}_p$ is $\E[m] = n p = \frac{n}{10(10 s)^{d-2}}$. If $S < 2^{\frac{1}{10} n^{1/(d-1)}}$, then $10 s < 2 n^{1/(d-1)}$, so $\E[m] > 2 s$, and by a Chernoff bound $\Pr[m > s] > 1/2$. Thus there exists a restriction $\rho$ for which simultaneously $m > s$ and $C|_\rho$ has depth $2$ and bottom fan-in $\le s$, contradicting Proposition~\ref{prop:depth2-parity}.
\end{proof}

\subsection{Matching Depth-\texorpdfstring{$d$}{d} Upper Bound and Logarithmic-Depth \texorpdfstring{$\NC^1$}{NC1} Construction}

H{\aa}stad's lower bound in Theorem~\ref{thm:hastad-parity-lb} is tight in the exponent $n^{1/(d-1)}$ for every fixed depth $d \ge 2$, and collapses to linear size $2n - 1$ once the depth reaches $\lceil \log_2 n \rceil$:

\begin{theorem}[Matching Depth-$d$ and $\NC^1$ Upper Bounds for Parity \citep{boppana1990circuit,hastad1987computational}]
\label{thm:parity-ub}
Let $n \ge 2$.
\begin{enumerate}[label=(\alph*),topsep=2pt,itemsep=2pt]
  \item For every integer $d \ge 2$, $\PAR_n$ can be computed by a depth-$d$ unbounded fan-in circuit over $\{\land, \lor, \neg\}$ of size at most $n \, 2^{n^{1/(d-1)}}$.
  \item At depth $O(\log n)$, $\PAR_n$ can be computed by a bounded fan-in circuit in $\NC^1$ of depth $2\lceil \log_2 n \rceil$ and size at most $3(n - 1)$ (or an unbounded fan-in alternating circuit of depth $\lceil \log_2 n \rceil + 1$ and size $2n - 1$).
\end{enumerate}
\end{theorem}
\begin{proof}
For part~(a), we proceed by induction on $d \ge 2$. At $d = 2$, $\PAR_n$ is a $2^{n-1}$-term DNF of size $2^{n-1} + 1 \le n 2^n$. For $d \ge 3$, set $m = \lceil n^{1/(d-1)} \rceil$ and partition the $n$ input bits into $b = \lceil n / m \rceil \le n^{(d-2)/(d-1)}$ blocks $y^{(1)}, \dots, y^{(b)}$ of at most $m$ variables each, so that
\[
  \PAR_n(x_1, \dots, x_n) = \PAR_b\bigl(\PAR_m(y^{(1)}), \dots, \PAR_m(y^{(b)})\bigr).
\]
By the induction hypothesis at depth $d - 1$, $\PAR_b$ can be computed by a depth-$(d-1)$ circuit $C_{\mathrm{out}}$ of size at most $b 2^{b^{1/(d-2)}} = b 2^{n^{1/(d-1)}}$. Depending on whether the bottom gates of $C_{\mathrm{out}}$ are $\lor$ or $\land$ gates, express each inner block $\PAR_m(y^{(r)})$ (and its negation) as an $m$-CNF ($\land$ of $\lor$'s) or $m$-DNF ($\lor$ of $\land$'s) of size $2^{m-1} + 1 = O(2^{n^{1/(d-1)}})$. Merging the top layer of each depth-$2$ block circuit with the bottom layer of $C_{\mathrm{out}}$ preserves depth $(d - 1) + 2 - 1 = d$ and yields total size at most $n 2^{n^{1/(d-1)}}$.

For part~(b), arrange $x_1, \dots, x_n$ at the leaves of a balanced binary tree of depth $\lceil \log_2 n \rceil$ with $n - 1$ internal $\mathrm{XOR}_2$ nodes. Since $\mathrm{XOR}_2(u, v) = (u \land \bar{v}) \lor (\bar{u} \land v) = (u \lor v) \land (\bar{u} \lor \bar{v})$, alternating CNF and DNF representations of $\mathrm{XOR}_2$ along adjacent levels of the binary tree and merging adjacent like-type gates yields depth $\lceil \log_2 n \rceil + 1$ and size $2(n - 1) + 1 = 2n - 1$.
\end{proof}

\section{The Razborov--Smolensky Polynomial Method for \texorpdfstring{$\AC^0[p]$}{AC0[p]}}
\label{sec:polynomial}

While the method of random restrictions settles the complexity of $\PAR_n$ in $\AC^0$, it breaks down immediately if the circuit is allowed unbounded fan-in $\MOD_3$ gates: restricting all but three variables of a $\MOD_3$ gate leaves an unsimplified $\MOD_3$ function on the remaining stars, which still requires exponential-size CNFs and DNFs. To prove lower bounds against $\AC^0[p]$ circuits for a prime $p$, Razborov \citep{razborov1987lower} and Smolensky \citep{smolensky1987algebraic} developed an algebraic approximation method over $\F_p$.

\subsection{Low-Degree Probabilistic Polynomials over \texorpdfstring{$\F_p$}{Fp}}

\begin{lemma}[Unique Multilinear Representation over $\F_p$]
\label{lem:fp-unique}
For any prime $p$, every function $f : \F_p^n \to \F_p$ is represented by a unique polynomial $P_f \in \F_p[X_1, \dots, X_n]$ of individual degree at most $p - 1$ in each variable $X_i$. Restricting the domain to the Boolean hypercube $\{0, 1\}^n \subseteq \F_p^n$, every Boolean function $f : \{0, 1\}^n \to \{0, 1\}$ has a unique multilinear representation $P_f(X_1, \dots, X_n) = \sum_{S \subseteq [n]} c_S \prod_{i \in S} X_i$ over $\F_p$.
\end{lemma}
\begin{proof}
For each $a = (a_1, \dots, a_n) \in \F_p^n$, the Lagrange indicator polynomial $\mathbf{1}_a(X) = \prod_{i=1}^n \prod_{z \in \F_p \setminus \{a_i\}} \frac{X_i - z}{a_i - z}$ has degree $p - 1$ in each variable and satisfies $\mathbf{1}_a(b) = \delta_{a, b}$ for all $b \in \F_p^n$. Thus $f(X) = \sum_{a \in \F_p^n} f(a) \mathbf{1}_a(X)$ spans the $p^n$-dimensional space of functions $\F_p^n \to \F_p$. Since there are exactly $p^n$ monomials $\prod_{i=1}^n X_i^{d_i}$ with $0 \le d_i \le p - 1$, linear independence and uniqueness follow by dimension counting. On $\{0, 1\}^n$, $X_i^2 = X_i$ reduces individual degrees to at most $1$, giving $2^n$ multilinear monomials matching $|\F_p^{\{0,1\}^n}| = p^{2^n}$.
\end{proof}

The key insight of Razborov \citep{razborov1987lower} and Smolensky \citep{smolensky1987algebraic} is that over $\F_p$, Fermat's Little Theorem implies $z^{p-1} = 1$ for every $z \in \F_p \setminus \{0\}$ and $0^{p-1} = 0$. Therefore, an unbounded fan-in $\MOD_p$ gate on inputs $f_1, \dots, f_t \in \{0, 1\}$ is computed \emph{exactly} by a polynomial of degree $p - 1$:
\begin{equation}
  \MOD_p(f_1, \dots, f_t) \;=\; \left(\sum_{i=1}^t f_i\right)^{\!p-1} \pmod p,
  \label{eq:modp-exact}
\end{equation}
whereas an unbounded fan-in $\lor$ or $\land$ gate can be approximated to error $2^{-\ell}$ by taking random $\F_p$-linear combinations of its inputs:

\begin{lemma}[Razborov--Smolensky Gate and Circuit Approximation \citep{razborov1987lower,smolensky1987algebraic}]
\label{lem:rs-approx}
Let $p$ be a prime and let $C$ be an $\AC^0[p]$ circuit of depth $d$ and size $S$ on $n$ variables. For every integer parameter $\ell \ge 1$, there exists a distribution over multilinear polynomials $\widetilde{P} \in \F_p[x_1, \dots, x_n]$ of total degree
\begin{equation}
  \deg(\widetilde{P}) \;\le\; \bigl((p - 1)\ell\bigr)^d
  \label{eq:rs-degree}
\end{equation}
such that for every input $x \in \{0, 1\}^n$, $\Pr_{\widetilde{P}}[\widetilde{P}(x) \neq C(x)] \le S \cdot 2^{-\ell}$. In particular, by averaging, there exists a deterministic polynomial $P \in \F_p[x_1, \dots, x_n]$ of degree at most $((p - 1)\ell)^d$ that agrees with $C(x)$ on at least a $1 - S \cdot 2^{-\ell}$ fraction of all inputs $x \in \{0, 1\}^n$.
\end{lemma}
\begin{proof}
We construct $\widetilde{P}$ inductively gate by gate from the inputs to the output:
\begin{enumerate}[label=(\arabic*),topsep=2pt,itemsep=2pt]
  \item \textbf{Literals and Negations:} Input $x_i$ is represented exactly by $x_i$ (degree $1$), and $\neg g$ is represented exactly by $1 - \tilde{g}$ (preserving degree).
  \item \textbf{$\MOD_p$ gates:} Given sub-circuit approximations $\tilde{f}_1, \dots, \tilde{f}_t$, we apply~\eqref{eq:modp-exact} to represent $\MOD_p(f_1, \dots, f_t)$ by $(\sum_{i=1}^t \tilde{f}_i)^{p-1} \pmod p$, multiplying the degree by at most $p - 1$ and introducing zero additional error whenever $\tilde{f}_i(x) = f_i(x)$ for all $i$.
  \item \textbf{Unbounded fan-in $\lor$ and $\land$ gates:} Consider $g = \bigvee_{i=1}^t f_i$ (where conjunctions $\bigwedge_{i=1}^t f_i = 1 - \bigvee_{i=1}^t (1 - f_i)$ follow by De Morgan's law). Draw $\ell$ independent uniformly random vectors $r^{(1)}, \dots, r^{(\ell)} \sim \F_p^t$ and define
  \begin{equation}
    \tilde{g}(x) \;:=\; 1 - \prod_{k=1}^\ell \left(1 - \Biggl(\sum_{i=1}^t r_i^{(k)} \tilde{f}_i(x)\Biggr)^{\!p-1}\right) \pmod p.
    \label{eq:rs-or-gate}
  \end{equation}
  Observe that if $\tilde{f}_i(x) = 0$ for all $i \in \{1, \dots, t\}$, then every linear combination $\sum_{i=1}^t r_i^{(k)} \tilde{f}_i(x) = 0$, so $\tilde{g}(x) = 0 = g(x)$ with probability $1$. Conversely, if $\tilde{f}_j(x) = 1$ for at least one index $j$, then for each $k \in \{1, \dots, \ell\}$ conditional on $(r_i^{(k)})_{i \neq j}$, the linear combination $\sum_{i=1}^t r_i^{(k)} \tilde{f}_i(x) \in \F_p$ vanishes for at most one choice of $r_j^{(k)} \in \F_p$, hence equals zero with probability $1/p \le 1/2$. By Fermat's Little Theorem, $(\sum_i r_i^{(k)} \tilde{f}_i(x))^{p-1} = 1$ whenever the sum is non-zero in $\F_p$, so $\Pr[\tilde{g}(x) \neq 1] \le (1/p)^\ell \le 2^{-\ell}$, and $\deg(\tilde{g}) \le (p - 1)\ell \max_i \deg(\tilde{f}_i)$.
\end{enumerate}
Accumulating degree blowup across $d$ layers yields $\deg(\widetilde{P}) \le ((p - 1)\ell)^d$, and taking a union bound over the at most $S$ gates of $C$ proves $\Pr[\widetilde{P}(x) \neq C(x)] \le S \cdot 2^{-\ell}$ for every $x \in \{0, 1\}^n$.
\end{proof}

\subsection{Degree Lower Bound and Exponential \texorpdfstring{$\AC^0[3]$}{AC0[3]} Lower Bound for Parity}

To complete the Razborov--Smolensky lower bound, we prove that no polynomial of degree $O(\sqrt{n})$ over $\F_3$ can agree with $\PAR_n$ on more than a $1/2 + O(d / \sqrt{n})$ fraction of the Boolean hypercube.

\begin{theorem}[Smolensky's Degree-Approximation Bound for Parity over $\F_3$ \citep{smolensky1987algebraic}]
\label{thm:smolensky-degree}
Let $P \in \F_3[x_1, \dots, x_n]$ be any polynomial of total degree $D \ge 1$. Then the fraction of inputs $x \in \{0, 1\}^n$ on which $P(x) = \PAR_n(x)$ satisfies
\begin{equation}
  \Pr_{x \sim \{0, 1\}^n}\bigl[P(x) = \PAR_n(x)\bigr]
  \;\le\;
  \frac{1}{2} + O\!\left(\frac{D}{\sqrt{n}}\right).
  \label{eq:smolensky-corr}
\end{equation}
\end{theorem}
\begin{proof}
Apply the invertible affine change of variables $y_i = 1 + x_i \pmod 3$, which maps $x_i \in \{0, 1\} \subseteq \F_3$ bijectively onto $y_i \in \{1, 2\} = \{+1, -1\} \subseteq \F_3 \setminus \{0\}$ (since $2 \equiv -1 \pmod 3$). Under this bijection, $\PAR_n(x) \in \{0, 1\}$ transforms into the full product monomial
\[
  \prod_{i=1}^n y_i = (-1)^{\sum_{i=1}^n x_i} = 1 - 2\,\PAR_n(x) \equiv 1 + \PAR_n(x) \pmod 3,
\]
and $P(x)$ transforms into a polynomial $Q(y_1, \dots, y_n) := 1 + P(y_1 - 1, \dots, y_n - 1) \in \F_3[y_1, \dots, y_n]$ of the same degree $\deg(Q) \le D$. Let
\[
  A := \Bigl\{y \in \{-1, +1\}^n : Q(y) = \prod_{i=1}^n y_i\Bigr\}
\]
denote the agreement set on which $P(x) = \PAR_n(x)$. We bound $|A|$ by computing the dimension of the vector space $\F_3^A$ of all functions $\phi : A \to \F_3$.

Notice that on $\{-1, +1\}^n$, every coordinate satisfies $y_i^2 = 1$ in $\F_3$, so every function $\phi : A \to \F_3$ is computed by a multilinear polynomial $\sum_{I \subseteq [n]} c_I \prod_{i \in I} y_i$. Moreover, for any $y \in A$ and any subset $I \subseteq [n]$ of size $|I| > n/2$, because $y_i^2 = 1$ and $\prod_{i=1}^n y_i = Q(y)$ on $A$, we can rewrite the high-degree monomial $\prod_{i \in I} y_i$ on $A$ as
\begin{equation}
  \prod_{i \in I} y_i
  \;=\;
  \Biggl(\prod_{i=1}^n y_i\Biggr) \Biggl(\prod_{j \notin I} y_j\Biggr)
  \;=\;
  Q(y) \prod_{j \notin I} y_j,
  \label{eq:monomial-degree-reduction}
\end{equation}
which is a polynomial of total degree at most $\deg(Q) + (n - |I|) \le D + n/2$! Consequently, every function $\phi \in \F_3^A$ can be represented on $A$ as a linear combination of multilinear monomials of degree at most $\lfloor n/2 \rfloor + D$. Because $\dim_{\F_3}(\F_3^A) = |A|$, the number of such monomials must be at least $|A|$:
\[
  |A|
  \;\le\;
  \sum_{k=0}^{\lfloor n/2 \rfloor + D} \binom{n}{k}
  \;=\;
  \sum_{k=0}^{\lfloor n/2 \rfloor} \binom{n}{k} + \sum_{k=\lfloor n/2 \rfloor + 1}^{\lfloor n/2 \rfloor + D} \binom{n}{k}
  \;\le\;
  \frac{2^n}{2} + D \binom{n}{\lfloor n/2 \rfloor}.
\]
By Stirling's formula, $\binom{n}{\lfloor n/2 \rfloor} \le \sqrt{\frac{2}{\pi n}} 2^n = O(2^n / \sqrt{n})$, so $|A| / 2^n \le \frac{1}{2} + O(D / \sqrt{n})$.
\end{proof}

\begin{corollary}[{Exponential Lower Bound for Parity Against $\AC^0[3]$ and $\MOD_q \notin \AC^0[p]$ \citep{razborov1987lower,smolensky1987algebraic}}]
\label{cor:ac0p-lb}
Every depth-$d$ $\AC^0[3]$ circuit $C$ computing $\PAR_n$ has size at least
\begin{equation}
  \operatorname{size}(C) \;\ge\; 2^{\Omega\bigl(n^{1/(2d)}\bigr)}.
  \label{eq:ac03-lb}
\end{equation}
More generally, for any two distinct primes $p \neq q$, the function $\MOD_q$ requires depth-$d$ $\AC^0[p]$ circuits of size $2^{\Omega(n^{1/(2d)})}$, so $\MOD_q \notin \AC^0[p]$.
\end{corollary}
\begin{proof}
Suppose $C$ is a depth-$d$ $\AC^0[3]$ circuit of size $S$ computing $\PAR_n$. Applying Lemma~\ref{lem:rs-approx} with $p = 3$ and $\ell = \lceil \log_2(4S) \rceil$ produces a polynomial $P \in \F_3[x_1, \dots, x_n]$ of degree $D \le (2\ell)^d = O((\log S)^d)$ that agrees with $\PAR_n$ on at least a $1 - S \cdot 2^{-\ell} \ge 3/4$ fraction of $\{0, 1\}^n$. By Theorem~\ref{thm:smolensky-degree}, any such polynomial must have degree $D \ge \Omega(\sqrt{n})$. Hence $(\log S)^d \ge \Omega(\sqrt{n})$, which implies $S \ge 2^{\Omega(n^{1/(2d)})}$. For general distinct primes $p \neq q$, replacing $\{-1, +1\} \subseteq \F_3$ by a primitive $q$-th root of unity $\omega \in \F_{p^r}$ replicates the monomial degree reduction~\eqref{eq:monomial-degree-reduction} identically \citep{smolensky1987algebraic}.
\end{proof}

\begin{table}[t]
\centering
\caption{Comparison of the two foundational lower-bound paradigms for bounded-depth Boolean circuits.}
\label{tab:methods-comparison}
\scriptsize
\setlength{\tabcolsep}{2.2pt}
\begin{tabular}{@{}lll@{}}
\toprule
Property & Random Restrictions (H{\aa}stad \citep{hastad1986almost}) & Polynomial Method \citep{razborov1987lower,smolensky1987algebraic} \\
\midrule
Target circuit class & $\AC^0$ (unbounded $\land, \lor, \neg$) & $\AC^0[p]$ for prime $p$ ($\land, \lor, \neg, \MOD_p$) \\
Analytic surrogate & Shallow decision trees ($\DT(f|_\rho) < s$) & Low-degree polynomials ($\deg(P) \le ((p-1)\ell)^d$) \\
Inductive step & Bottom-layer CNF$\leftrightarrow$DNF switching & Gate-by-gate random $\F_p$-linear projection \\
Hardness witness & $\PAR_n$ requires $\DT(\PAR_m) = m$ & $\MOD_q$ ($q \neq p$) spans $\Omega(2^n)$ monomials \\
Depth-$d$ size bound & $2^{\Omega(n^{1/(d-1)})}$ (tight in exponent) & $2^{\Omega(n^{1/(2d)})}$ \\
Main limitation & Fails on $\MOD_p$ gates & Fails on composite $\MOD_m$ ($m = pq$, $\ACC^0$) \\
\bottomrule
\end{tabular}
\end{table}

\section{Structured Restrictions, Consequences, and Barriers}
\label{sec:consequences}

We conclude by examining how non-uniform and structured random restrictions extend H{\aa}stad's Switching Lemma beyond parity, how the Switching Lemma drives core results in Fourier analysis, pseudorandomness, and circuit satisfiability, and why the Razborov--Rudich Natural Proofs barrier limits its reach above $\AC^0$.

\subsection{Biased, Graph-Structured, and Sunflower Switching Lemmas}

While the symmetric restriction $\mathcal{R}_p = \mathcal{R}(p, 1/2)$ assigns non-star variables to $0$ and $1$ with equal probability, monotone graph properties such as $k\text{-}\CLIQ_n$ on $\binom{n}{2}$ edge variables become constant under $\mathcal{R}(p, 1/2)$ unless edges are set to $1$ with very small probability $q$ and to $0$ with probability $1 - q$. Beame \citep{beame1994switching}, Rossman \citep{rossman2008constant}, and Rossman, Servedio, and Tan \citep{rossman2015average} developed structured switching lemmas tailored to these distributions:
\begin{enumerate}[label=(\arabic*),topsep=2pt,itemsep=2pt]
  \item \textbf{$q$-Biased Switching Lemma \citep{beame1994switching}:}
  Under the biased product distribution $\mathcal{R}(p, q)$ of~\eqref{eq:Rpq-def} with $0 < q \le 1/2$, any $k$-DNF $f$ satisfies
  \begin{equation}
    \Pr_{\rho \sim \mathcal{R}(p, q)}\bigl[\DT(f|_\rho) \ge s\bigr] \;\le\; \left(O\!\left(\frac{p k}{q}\right)\right)^{\!s},
    \label{eq:biased-switching}
  \end{equation}
  enabling average-case $\AC^0$ lower bounds under $q$-biased product measures on $\{0, 1\}^n$.
  \item \textbf{Clique-Preserving Restrictions and $k$-Clique Lower Bounds \citep{beame1994switching,rossman2008constant}:}
  To preserve the boundary between having and not having a $k$-clique on an $n$-vertex graph, one draws a random vertex subset $V_* \subseteq [n]$ of survival density $p$, keeps edges inside $\binom{V_*}{2}$ as $*$ variables, and sets edges outside $\binom{V_*}{2}$ with bias $q$. The resulting clique switching lemma yields an $n^{\Omega(k / d^2)}$ lower bound for $k\text{-}\CLIQ_n$ on depth-$d$ $\AC^0$ circuits \citep{beame1994switching}, and Rossman's entropy-based method \citep{rossman2008constant} establishes the depth-independent average-case lower bound $\omega(n^{k/4})$ for constant-depth circuits computing $k\text{-}\CLIQ_n$.
  \item \textbf{Random Projections and Average-Case Depth Hierarchy \citep{rossman2015average}:}
  Replacing independent random restrictions with \emph{random projections} (which identify groups of surviving star variables to maintain correlation structure across layers), Rossman, Servedio, and Tan \citep{rossman2015average} proved an optimal average-case depth hierarchy theorem: for every $d \ge 2$, there is an explicit monotone function in $\AC^0$ of depth $d$ and linear size that has correlation $n^{-\Omega(1/d)}$ with any depth-$(d-1)$ circuit of size $2^{n^{\Omega(1/d)}}$.
  \item \textbf{Sunflower-Based Switching Lemmas \citep{alweiss2021improved}:}
  Alweiss, Lovett, Wu, and Zhang \citep{alweiss2021improved} established a breakthrough $(\log k)^{O(k)}$ bound for the Erd\H{o}s--Rado Sunflower Lemma by proving a structural switching lemma that reduces any $k$-uniform set system of low spread to a shallow decision tree under random restrictions.
\end{enumerate}

\subsection{Four Foundational Consequences of \texorpdfstring{$\PAR \notin \AC^0$}{PAR not in AC0}}

H{\aa}stad's Switching Lemma (Theorem~\ref{thm:switching-lemma}) implies that every $\AC^0$ circuit simplifies to a shallow decision tree under a random restriction. This structural collapse has four major consequences across complexity theory:

\begin{enumerate}[label=(\arabic*),topsep=3pt,itemsep=3pt]
  \item \textbf{Oracle Separation of $\PH$ from $\PSPACE$ \citep{furst1984parity,yao1985separating,hastad1987computational}:}
  For any oracle $A \subseteq \{0, 1\}^*$, converting a constant-alternation $\PH^A$ machine on input length $n$ into an alternating quantifier tree produces an $\AC^0$ circuit of depth $d = O(1)$ and size $2^{\operatorname{poly}(\log n)}$ querying bits of the truth table of $A$ on inputs of length $m = \operatorname{poly}(\log n)$. By Theorem~\ref{thm:hastad-parity-lb}, since $2^{\operatorname{poly}(\log n)} = 2^{m^{o(1)}} \ll 2^{\Omega(m^{1/(d-1)})}$, no such circuit can compute the parity of $2^m$ oracle bits, whereas a $\PSPACE^A$ machine can scan all $2^m$ strings of length $m$ in polynomial space. Diagonalizing across input lengths yields an oracle $A$ relative to which $\PH^A \subsetneq \PSPACE^A$, proving that no relativizing technique can resolve $\PH \stackrel{?}{=} \PSPACE$.

  \item \textbf{Fourier Tail Concentration and Learnability of $\AC^0$ \citep{linial1993constant,tal2017tight}:}
  Identifying $\{0, 1\}^n$ with $\{-1, +1\}^n$ via $\chi_S(x) = \prod_{i \in S} (-1)^{x_i}$, every function $f : \{0, 1\}^n \to \{-1, +1\}$ has the orthonormal Fourier expansion $f(x) = \sum_{S \subseteq [n]} \hat{f}(S) \chi_S(x)$ with $\sum_{S} \hat{f}(S)^2 = 1$. If $\DT(f|_\rho) < s$, then $f|_\rho$ has degree less than $s$ as a real polynomial; since a subset $S \subseteq [n]$ of size $|S| = k$ survives inside $\rho^{-1}(*)$ with probability $p^k$, the Switching Lemma implies the Linial--Mansour--Nisan tail bound \citep{linial1993constant} (sharpened to its optimal base by Tal \citep{tal2017tight}): for any depth-$d$ size-$M$ circuit $f \in \AC^0$,
  \begin{equation}
    \sum_{S \subseteq [n] : |S| \ge k} \hat{f}(S)^2
    \;\le\;
    2 \, M \, 2^{-k / O(\log M)^{d-1}}.
    \label{eq:lmn-tail}
  \end{equation}
  Consequently, estimating all low-degree Fourier coefficients $|S| \le k_0$, where $k_0 = O(\log(M/\varepsilon))^{d-1} \log(1/\varepsilon)$, from random examples learns any $\AC^0$ function under the uniform distribution in quasipolynomial time $n^{O(k_0)}$.

  \item \textbf{Pseudorandom Generators and Polylogarithmic Independence for $\AC^0$ \citep{nisan1994hardness,braverman2010polylogarithmic}:}
  Using the hardness-versus-randomness paradigm \citep{nisan1994hardness} and the average-case hardness of $\PAR_n$ against $\AC^0$ \citep{hastad1987computational,hastad2014correlation}, Nisan and Wigderson constructed an unconditional pseudorandom generator $G$ stretching $\operatorname{polylog}(n)$ seed bits to $n$ bits that $\varepsilon$-fools all polynomial-size $\AC^0$ circuits, derandomizing $\mathsf{BPAC}^0$ in polylogarithmic space and quasipolynomial time. Furthermore, Braverman \citep{braverman2010polylogarithmic} proved via the Switching Lemma and sandwiching polynomials that any $(\log(M/\varepsilon))^{O(d^2)}$-wise independent distribution $\varepsilon$-fools $\AC^0$ circuits of size $M$ and depth $d$.

  \item \textbf{Subexponential $\AC^0$-SAT Algorithms and the Algorithmic Method for $\ACC^0$ \citep{impagliazzo2012satisfiability,williams2014nonuniform}:}
  Given a depth-$d$ circuit $C$ of size $S$ on $n$ variables, Impagliazzo, Matthews, and Paturi \citep{impagliazzo2012satisfiability} showed that partitioning the $n$ variables into deterministic blocks of size $1/p = O(\log S)^{d-1}$ to simulate the Switching Lemma reduces $C$ to a family of shallow decision trees, solving $\AC^0\text{-}\mathrm{SAT}$ and $\#\mathrm{SAT}$ in deterministic time $2^{n (1 - 1 / O(\log S)^{d-1})}$. Conversely, Williams \citep{williams2014nonuniform} proved a foundational meta-theorem connecting faster-than-$2^n$ circuit satisfiability algorithms to non-uniform circuit lower bounds, using a $2^{n - n^\varepsilon}$ satisfiability algorithm for $\ACC^0$ circuits (via symmetric representations of Grothendieck-like polynomial approximations) to establish the breakthrough separation $\NEXP \not\subseteq \ACC^0$.
\end{enumerate}

\subsection{The Razborov--Rudich Natural Proofs Barrier}

Why have neither random restrictions nor algebraic approximations succeeded in proving super-polynomial lower bounds for $\mathsf{TC}^0$ (bounded-depth threshold circuits), $\NC^1$, or general $\Ppoly$ circuits? Razborov and Rudich \citep{razborov1997natural} proved that any combinatorial lower-bound technique satisfying two structural properties inherent to both the Switching Lemma and the Razborov--Smolensky method would automatically break cryptographic pseudorandom function generators (PRFs):

\begin{theorem}[Natural Proofs Barrier \citep{razborov1997natural}]
\label{thm:natural-proofs}
A combinatorial property $\mathcal{C}_n \subseteq \{f : \{0, 1\}^n \to \{0, 1\}\}$ of $n$-variable Boolean functions is \emph{$\Gamma$-natural against $\Lambda$} if:
\begin{enumerate}[label=(\roman*),topsep=2pt,itemsep=2pt]
  \item \textbf{Usefulness against $\Lambda$:} Every circuit family $\{C_n\} \in \Lambda$ eventually satisfies $C_n \notin \mathcal{C}_n$.
  \item \textbf{Constructivity in $\Gamma$:} Given the $N = 2^n$-bit truth table of $f : \{0, 1\}^n \to \{0, 1\}$, deciding whether $f \in \mathcal{C}_n$ is computable in complexity class $\Gamma$ (e.g., $\Gamma = \mathsf{P}$ in the input size $N = 2^n$).
  \item \textbf{Largeness:} A uniformly random Boolean function $f : \{0, 1\}^n \to \{0, 1\}$ satisfies $\Pr_f[f \in \mathcal{C}_n] \ge 2^{-O(n)}$.
\end{enumerate}
If there exists a $\Ppoly$-natural proof against a circuit class $\Lambda$ that is strong enough to compute a pseudorandom function ensemble against $2^{O(n)}$-size adversaries (which exists inside $\mathsf{TC}^0$ and $\NC^1$ under standard cryptographic hardness assumptions such as the hardness of factoring or discrete logarithms), then no such pseudorandom function generator exists in $\Lambda$.
\end{theorem}

Both H{\aa}stad's random restriction method and the Razborov--Smolensky polynomial method define $\mathsf{P}$-natural properties: for example, the property $\mathcal{C}_n$ consisting of functions $f$ for which some restriction $\rho \in \mathcal{R}_t$ satisfies $\DT(f|_\rho) \ge s$ is decidable in time $\operatorname{poly}(2^n)$ from the truth table of $f$ and holds for almost all random Boolean functions (largeness), while failing for every $f \in \AC^0$ (usefulness). These proofs succeed for $\AC^0$ and $\AC^0[p]$ precisely because $\AC^0$ and $\AC^0[p]$ cannot compute pseudorandom functions (indeed, by the Linial--Mansour--Nisan Fourier concentration bound~\eqref{eq:lmn-tail}, any $\AC^0$ function can be distinguished from a random function in quasipolynomial time $2^{\operatorname{polylog}(n)} \ll 2^n$). Overcoming the Natural Proofs barrier for $\mathsf{TC}^0$, $\NC^1$, and $\Ppoly$ requires non-natural techniques---such as Williams's algorithmic method \citep{williams2014nonuniform}, which exploits uniform $\NEXP$ diagonalization and avoids the largeness condition by analyzing satisfiability on succinct circuit encodings rather than dense truth tables.

\bibliographystyle{unsrtnat}
\bibliography{refs}

\end{document}